\documentclass[12pt,a4paper]{article}
\usepackage[a4paper, portrait, margin=0.5in]{geometry}

\usepackage[utf8]{inputenc}
\usepackage[T1]{fontenc}  

\usepackage[english]{babel}
\usepackage{graphicx}
\graphicspath{ {./Images/} }

\usepackage{amsmath}
\usepackage{amssymb}
\usepackage{url}
\usepackage{mathtools}
\usepackage{graphics,graphicx,color,float}
\usepackage{float}
\usepackage{etoolbox} 
\usepackage{amsfonts}
\usepackage{amsthm}
\usepackage[usenames,dvipsnames]{xcolor}
\usepackage{tikz}
\usepackage[T1]{fontenc}
\usetikzlibrary{matrix}
\usepackage{xr-hyper}
\usepackage{tikz}
\usetikzlibrary{calc}

\newcommand{\la}{\langle}
\newcommand{\ra}{\rangle}

\newcommand{\R}{\mathbb{R}}

\def\bpm{\begin{pmatrix}}
\def\epm{\end{pmatrix}}

\newcommand{\beq}{\begin{equation}}
\newcommand{\enq}{\end{equation}}
\newcommand{\bel}{\begin{lemma}}

\newcommand{\bet}{\begin{theorem}}
\newcommand{\ent}{\end{theorem}}

\newcommand{\tr}{\mathrm{tr}}
\makeatletter
\newcommand*{\addFileDependency}[1]{
  \typeout{(#1)}
  \@addtofilelist{#1}
  \IfFileExists{#1}{}{\typeout{No file #1.}}
}
\makeatother

\newcommand{\Tr}{\mathrm{tr}}
\newcommand{\ketbra}[1]{|#1\rangle\langle#1|}

\usepackage{booktabs,tabularx}

\newcommand{\suppress}[1]{}

\newcommand{\bra}[1]{\langle #1|}
\newcommand{\ket}[1]{|#1 \rangle}
\newcommand{\braket}[2]{\langle #1|#2\rangle}

\def\be{\begin{equation}}
\def\ee{\end{equation}}

\newcommand{\gexcl}{\mathcal{G}_{\mathrm{ex}}}

\makeatletter
\newcommand*{\rom}[1]{\expandafter\@slowromancap\romannumeral #1@}
\makeatother

\makeatletter
\appto{\appendix}{%
 \@ifstar{\def\theequation@prefix{A.}}%
 {}%
}
\makeatother

\mathchardef\mhyphen="2D

\usepackage{amsmath}
\usepackage{amssymb}
\usepackage{url}
\usepackage{mathtools}
\usepackage{graphics,graphicx,color,float}
\usepackage{etoolbox} 
\usepackage{amsfonts}

\usepackage[usenames,dvipsnames]{xcolor}
\usepackage{tikz}
\usepackage[T1]{fontenc}
\usetikzlibrary{matrix}

\usepackage{tikz}
\usetikzlibrary{calc}

\makeatletter
\appto{\appendix}{%
 \@ifstar{\def\theequation@prefix{A.}}%
 {}%
}
\makeatother

\newtheorem{theorem}{Theorem}[section]
\newtheorem{lemma}[theorem]{Lemma}
\newtheorem{corollary}[theorem]{Corollary}
\newtheorem{definition}[theorem]{Definition}
\newtheorem*{definition*}{Definition}
\newtheorem*{example*}{Example}
\newtheorem*{main result}{Main Result}

\usepackage{times}
\usepackage{fullpage}
\usepackage{amsfonts,amssymb,amsmath,amsthm}
\usepackage{latexsym} 
\usepackage{gauss}
\usepackage{tikz}

\usepackage{url}
\usepackage{thmtools}
\usepackage{thm-restate}

\newcommand{\T}{\mathrm{T}}

\def\01{\{0,1\}}

\usepackage{amsthm}

\usepackage[linesnumbered,ruled,procnumbered]{algorithm2e}

\usepackage{hyperref}
\hypersetup{
    colorlinks=true,
    citecolor = [rgb]{0.3,0.8,0},
    linkcolor=blue,
    filecolor=magenta,      
    urlcolor= magenta,
}

\usepackage[capitalise, nameinlink]{cleveref}

\usepackage{graphicx}
\graphicspath{{images/}{../images/}}
\usepackage{framed}
\colorlet{shadecolor}{blue!11}

\usepackage[affil-it]{authblk} 
\usepackage{blindtext}
\usepackage[linesnumbered,ruled,procnumbered]{algorithm2e}
\usepackage{multicol}

\begin{document}

\title{ Graph-theoretic approach to dimension witnessing}
\author[1,2]{ \small Maharshi Ray \footnote{\href{mailto:maharshi91@gmail.com}{e-mail : \texttt{maharshi91@gmail.com}}} }
\author[1]{Naresh Goud Boddu}
\author[1]{Kishor Bharti}
\author[1,2,3]{Leong-Chuan Kwek}
\author[4,5]{Ad\'{a}n Cabello}

\affil[1]{ \footnotesize Centre for Quantum Technologies, National University of Singapore}
\affil[2]{\footnotesize MajuLab, CNRS-UNS-NUS-NTU International Joint Research Unit, Singapore UMI 3654, Singapore}
\affil[3]{\footnotesize National Institute of Education, Nanyang Technological University, Singapore 637616, Singapore}
\affil[4]{\footnotesize Departamento de F\'{i}sica Aplicada II, Universidad de Sevilla, E-41012 Sevilla, Spain}
\affil[5]{\footnotesize Instituto Carlos I de F\'{\i}sica Te\'orica y Computacional, Universidad de Sevilla, E-41012 Sevilla, Spain}

\date{\vspace{-5ex}} %
\maketitle


\begin{abstract}
A fundamental problem in quantum computation and quantum information is finding the minimum quantum dimension needed for a task. For tasks involving state preparation and measurements, this problem can be addressed using only the input-output correlations. This has been applied to Bell, prepare-and-measure, and Kochen-Specker contextuality scenarios. Here, we introduce a novel approach to quantum dimension witnessing for scenarios with one preparation and several measurements, which uses the graphs of mutual exclusivity between sets of measurement events. We present the concepts and tools needed for graph-theoretic quantum dimension witnessing and illustrate their use by identifying novel quantum dimension witnesses, including a family that can certify arbitrarily high quantum dimensions with few events.
\end{abstract}


\section{Introduction}


The dimensionality of a quantum system is crucial for its ability to perform quantum information processing tasks. For example, the security of some protocols for quantum key distribution and randomness expansion depends on the presumed dimensionality of the underlying physical system. The dimensionality also plays a crucial role in device characterisation tasks. Also, non-classical phenomena such as Kochen-Specker contextuality is known to require quantum systems of dimension at least three \cite{Kochen:1967JMM}. Therefore, it is of fundamental importance to have efficient tools to determine the dimensionality of the underlying Hilbert space where the measurement operators act on the physical system for any experimental setup.

There are several approaches to tackle this problem. One of them is known as {\em self-testing} \cite{Yao_self}. The idea of self-testing is to identify unique equivalence class of configurations corresponding to extremal quantum violation of a Bell inequality. The members of the equivalence class are related via some fixed local isometry. The dimension of the individual quantum system can be lower bounded by identifying the equivalence class of configurations attaining the optimality \cite{Yao_self}. Though initially proposed in the setting of Bell non-locality, the idea of self-testing has been extended to prepare-and-measure scenarios, contextuality, and quantum steering \cite{tavakoli2018self, BRVWCK19, bharti2019local,vsupic2016self,shrotriya2020self}. For a review of self-testing, we refer to \cite{vsupic2019self}. It is important to stress that only extremal points of the quantum set of correlations that can be attained via finite-dimensional configurations admit self-testing \cite{goh2018geometry}.

The second approach is {\em tomography}. Quantum tomography is a process via which the description of a quantum state is obtained by performing measurements on an ensemble of identical quantum states. For quantum systems of dimension $d$, to estimate an unknown quantum system to an error $\epsilon$ (in $l_1$ norm) requires $\Theta \left(d^2 \epsilon^{-2}\right)$ copies of a quantum state \cite{OW17}. One drawback of this approach is that it requires a prior knowledge of the dimensionality of the system.

The third approach is {\em dimension witnesses} \cite{brunner_testing_2008}. This is the approach we will focus on in this paper.
The goal of dimension witness is to render a lower bound on the dimensionality of the underlying physical system based on the experimental statistics. For example, a quantum dimension witness is a quantity that can be computed from the input-output correlations and whose value gives a lower bound to the dimension of the Hilbert space needed to accommodate the density matrices and the measurement operators needed to produce such correlations. Dimension witnesses have been investigated for the following types of scenarios:
\begin{enumerate}
\item \label{type1} \textbf{Bell scenarios:} Here, quantum dimension witnesses are based on the observation that certain bipartite Bell non-local correlations are impossible to produce with quantum systems of local dimension $d$ (and thus global dimension $d^2$) or less, implying that the experimental observation of these correlations certifies that the quantum local dimension is at least $d+1$ \cite{brunner_testing_2008,vertesi_bounding_2009,brunner_dimension_2013}. There are dimension witnesses of this type for arbitrarily high quantum local dimension $d$ \cite{brunner_testing_2008}, but they require preparing entangled states of dimension $d^2$ and conditions of spatial separation that do not occur naturally in quantum computers. This approach to dimension witnessing is related to self-testing based on Bell non-local correlations \cite{Yao_self}. A Bell dimension witness certifies the minimum quantum dimension accessed by the measurement devices acting on the physical systems prepared by a single source. 
\item \label{type2} \textbf{Prepare-and-measure scenarios:} These scenarios consists of $p$ different preparation sources and $m$ measurements acting on the physical systems emitted by those sources. Prepare-and-measure dimension witnesses require $p > d+1$ preparations to certify classical or quantum dimension $d$ \cite{wehner2008lower,gallego_device-independent_2010}. They have been used to experimentally certify in a device-independent way small classical and quantum dimensions \cite{hendrych_experimental_2012,ahrens2012experimental,d2014device}. A prepare-and-measure dimension witness certifies the minimum classical or quantum dimension spanned by the $p$ preparation sources and the $m$ measurements. 
\item \label{type3} \textbf{Kochen-Specker contextuality scenarios:} They consist of a single state preparation followed by a sequence of compatible ideal measurements chosen from a fixed set. Two measurements are compatible (or jointly measurable) when there is a third measurement that works as a refinement for both of them, so each of them can be measure by coarse graining the third measurement and thus both of them can be jointly measured. A measurement is ideal when it yields the same outcome when repeated on the same physical system and does not disturb any compatible measurement. Checking experimentally that a set of measurements are ideal and have certain relations of compatibility can be done from the input-output correlations \cite{LMZNCAH18}. Correlations between the outcomes of ideal measurements are Kochen-Specker contextual when they cannot be reproduced with models in which measurements have predetermined context-independent outcomes~\cite{cabello2008experimentally,KCBS}. Quantum Kochen-Specker contextuality dimension witnesses are based on the observation that certain Kochen-Specker contextual correlations are impossible to produce with quantum systems of dimension $d$ or less, implying that its experimental observation certifies a local quantum dimension of at least $d$. The problem of contextuality dimension witnesses is that they require testing in addition that the measurements are ideal and satisfy certain relations of compatibility. A {\em state-dependent} contextuality dimension witness certifies the minimum quantum dimension accessed by the measurement devices acting on the physical systems prepared by a single source. In a {\em state-independent} contextuality scenario, these measurements form a state-independent contextuality set in dimension $d$, defined as one for which the quantum predictions for sequences of compatible measurements for any quantum state in dimension $d$ cannot be reproduced by non-contextual models \cite{cabello2015necessary}. The minimum quantum dimension for contextual correlations have been studied in~\cite{GBCKL14}. A state-independent Kochen-Specker contextuality dimension witness certifies the minimum quantum dimension accessed by the measurement devices, without relating the conclusion to any particular source.
\end{enumerate}
                                                     
In this paper, we introduce a novel graph-theoretic approach to quantum dimension witnessing. We deal with abstract structures of measurement events produced for one preparation and several measurements, as is the case in Kochen-Specker contextuality and Bell scenarios. This means that our approach will always work in Kochen-Specker contextuality scenario and sometimes in specific Bell scenarios. 

Our approach is, first, based on the observation that the problem of finding dimension witnesses can be reformulated as the problem of finding correlations for structures of exclusivity which are impossible to produce with systems of quantum dimension $d$ or less, implying that its experimental observation certifies a quantum dimension of at least $d+1$. Second, it is based on the observation that, given a set of events and their relations of mutual exclusivity, the sets of correlations allowed in quantum theory are connected to well-known and easy to characterize invariants and sets in graph theory \cite{CSW}. In fact, the power of the graph-theoretic approach to dimension witnessing is based on three pillars:
\begin{itemize}  
\item The connection between correlations for structures of exclusivity and easy to characterize sets in graph theory. This connection allows us to use tools and results of graph theory for quantum graph dimension witnessing. 
\item The observation that finding dimension witnesses in scenarios with many measurements is difficult due to the difficulty to fully characterize in these scenarios the sets of correlations that cannot be achieved with a given dimension. In contrast, the graph approach allows us to rapidly identify structures of exclusivity that have dimension witnesses, even though many of them correspond to scenarios with many measurements. 
\item The connection between abstract structures of exclusivity and some specific contextuality scenarios (those consisting of dichotomic measurements having a structure of compatibility isomorphic to the structure of exclusivity). This assures that any quantum dimension witness for a graph of exclusivity always admits a physical realization in {\em some} Kochen-Specker contextuality scenario. Moreover, by imposing extra constraints, we can find, in principle, those dimension witness that also admit a physical realizations in a {\em specific} Kochen-Specker contextuality or Bell scenario.
\end{itemize}


The paper is organized as follows. In Sec.~\ref{notation_context} we introduce some standard definitions of graph theory and the graph-theoretic approach to correlations. In Sec.~\ref{sec2}, we use this graph-theoretic approach to study quantum dimension witness. Specifically, in Subsec.~\ref{heuristics}, we present a heuristic technique to compute a lower bound on the $d$ dimensional-restricted quantum value and find the corresponding $d$-dimensional quantum realisations. We illustrate the usefulness of this tool with some examples.
In Subsec.~\ref{sec4Qites}, we introduce a family of graphs, which we call the $k$-Qite family, and show that their elements are relatively simple quantum dimension witness for any dimension $k \geq 3$. Finally, in Sec.~\ref{disc}, we conclude by listing future directions for research.

Most of the notations used in the paper are self-explanatory. A graph describes relationships between several entities or vertices. We denote an edge between two vertices $i$ and $j$ by the symbol $i \sim j$. A class of commonly studied graphs is the cycles on $n$ vertices, which we denote by $C_n$. The work also uses semidefinite programming where we use the symbol $S_+^{n}$ to denote the class of positive semi-definite hermitian matrices of size $n \times n$.


\section{Graph theoretic approach to contextuality} 
\label{notation_context}


Consider an experiment in the black-box setting. 
An outcome $a$ and its associated measurement $M$, are together called a measurement event and denoted as $(a|M)$.

\begin{definition}(Exclusive event)
Two events $e_{i}$ and $e_{j}$ are defined to be exclusive if there exists a measurement $M$ such that $e_{i}$ and $e_{j}$ correspond to different outcomes of $M,$ i.e. $e_{i}=\left(a_{i} \mid M\right)$ and $e_{j}=\left(a_{j} \mid M\right)$ such that $a_{i} \neq a_{j}.$
\end{definition}

\begin{definition}(Exclusivity graph)
For a family of events $\left\{e_{1}, e_{2} \ldots e_{n}\right\}$ we associate a simple undirected graph, $\gexcl:=(V, E),$ with vertex set $V$ and edge set $E$ such that two vertices $i, j \in V$ share an edge if and only if $e_{i}$ and $e_{j}$ are exclusive events. $G$ is called an exclusivity graph.
\end{definition}

Now we consider theories that assign probabilities to the events corresponding to its vertices. Concretely, a {\em behaviour} corresponding to $\gexcl$ is a mapping $p: [n]\to [0,1]$, such that $p_i+p_j\le 1$, for all $i\sim j$, where we denote $p(i)$ by $p_i$. 
Here, the non-negative scalar $p_i\in [0,1]$ encodes the probability that measurement event $e_i$ occurs. Furthermore, note that two exclusive events $e_i$ and $e_j$ implies the linear constraint $p_i+p_j\le 1$.

A behaviour $p: [n]\to [0,1]$ is {\em deterministic non-contextual} if each $p_i \in \{0,1\}$ such that $p_i+p_j \leq 1$ for exclusive events $e_i$ and $e_j$. A {\em deterministic non-contextual} behaviour can be considered as a vector in $\mathbb{R}^n$. The polytope of {\em non-contextual behaviours}, denoted by $\mathcal{P}_{NC}(\gexcl)$, is the convex hull of all deterministic non-contextual behaviours. The behaviours that do not lie in $\mathcal{P}_{NC}(\gexcl)$ are called {\em contextual}. It is worthwhile to mention that in combinatorial optimisation, one often encounters the {\em stable set} polytope of a graph $G$, $STAB(G)$ (defined below). It is quite easy to see that stable sets of $G$ (a subset of vertices, where no two vertices share an edge between them) and {\em deterministic} behaviours coincide. 

\begin{definition}
\[ STAB(G) = \{ conv(x) : x \text{ is a characteristic vector of a stable set of } G \}
\] 
\end{definition}It thus follows from the definition that $\mathcal{P}_{NC}(\gexcl)=STAB(\gexcl)$. 

Lastly, a behaviour $p: [n]\to [0,1]$ is called {\em quantum} if there exists a quantum state $\ket{\psi}$ and projectors $\Pi_1,\ldots \Pi_n$ acting on a Hilbert space $\mathcal{H}$ such that 
\be p_i= \bra{\psi}\Pi_i \ket{\psi}, \forall i\in [n] \text{ and } \Tr(\Pi_i\Pi_j)=0, \text{ for } i\sim j.\ee 
We refer to the ensemble $\ket{\psi}, \{\Pi\}_{i=1}^n$ as a {\em quantum realization} of the behaviour $p$. 
The convex set of all quantum behaviours is denoted by $\mathcal{P}_{Q}(\gexcl)$. It turns out this set too is a well studied entity in combinatorial optimisation, namely the {\em theta body}. 

\begin{definition}
The theta body of a graph $G=([n],E)$ is defined by: 
$${\rm TH}(G)=\{x\in \R^n_+: \exists Y\in \mathbb{S}^{1+n}_+, \ Y_{00}=1, \ Y_{ii}=x_i = Y_{0i} \quad \, \forall i \in [n], \ Y_{ij}=0, \forall (i,j)\in E\}.$$
\end{definition}

The fact that $\mathcal{P}_{Q}(\gexcl) = TH(\gexcl)$, was observed in ~\cite{CSW} and follows by taking $d = \ket{\psi}$ and $w_i = \Pi_i \ket{\psi} /\sqrt{\bra{\psi}\Pi_i \ket{\psi}} \ \forall i \in [n]$, in the following lemma. 
\begin{lemma}\label{startpoint}
We have that $x\in TH(G)$ iff there exist unit vectors $d,w_1,\ldots,w_n$ such that 
\be\label{csdcever}
x_i=\la d,w_i\ra^2, \forall i\in [n] \text{ and } \la w_i, w_j\ra=0, \text{ for } (i,j)\in E.
\ee
\end{lemma}

\begin{proof}Let $x\in {\rm TH}(G)$. By definition, $x$ is the diagonal of a matrix $Y$ satisfying $Y\in \mathbb{S}^{1+n}_+, \ Y_{00}=1, \ Y_{ii}=Y_{0i}, \ Y_{ij}=0, \forall (i,j)\in E$. Let $Y=Gram(d,v_1,\ldots,v_n)$. Define $w_i={v_i\over \|v_i\|}$. Using that $x_i=Y_{ii}=Y_{0i}$ we get that 
$$x_i=\la v_i,v_i\ra=\la d,v_i\ra=\la d,w_i\|v_i\|\ra=\|v_i\|\la d,w_i\ra.$$
Lastly, note that $\la d,w_i\ra=\la d, {v_i\over \|v_i\|}\ra={ \la v_i,v_i\ra \over \|v_i\|}=\|v_i\|.$ Combining these two equations we get that 
$$x_i=\la d,w_i\ra^2.$$

\noindent Conversely, let $Y$ be the Gram matrix of $d, \la d,w_1\ra w_1,...,\la d,w_1\ra w_1$. Note that $\la d,w_i\ra w_i$ is the orthogonal projection of $d$ onto the unit vector $w_i$. It is easy to see that $Y$ has all the desired properties.
\end{proof} 

\noindent In the above lemma, the vectors $w_i$, for $i \in [n]$, are sometimes referred to as an orthonormal representation (OR) of $G$. 
\begin{definition}(orthonormal representation) An orthonormal representation of a graph $G = (V,E)$, is a set of unit vectors $w_i$ for $i \in [|V|]$, such that $\braket{w_i}{w_j} = 0, \text{ for all } (i,j) \in E$.
\end{definition}

\noindent The cost of this orthonormal representation of the graph is defined as $\lambda_{\max}\left( \sum_{i \in [|V|]} \ketbra{w_i}\right)$.

\medskip 

Next, we turn our attention to the sum $S = p_1 + p_2 + \cdots + p_n$, where $p \in \mathcal{P}_{NC}(\gexcl)$ is a {\em non-contextual} behaviour. The set of non-contextual behaviors forms a bounded polyhedron i.e. a polytope. The facets of the aforementioned polytope define tight non-contextuality inequalities, which correspond to half-spaces. This explains why we are interested in $\sum_i p_i $. The maximum of $S$ over {\em deterministic} behaviours is the same as the maximum of $S$ over {non-contextual} behaviours. To see this, let $p \in \mathcal{P}_{NC}(\gexcl)$ be a maximizer of $S$. We can write $p$ as a convex sum of deterministic behaviours, that is $p = \sum_j \lambda_j p^{(j)}$, where $p^{(j)}$ are deterministic behaviours and $\lambda_i > 0, \ \sum_i \lambda_i = 1$. Now, note that the optimal value of $S = \sum_j \lambda_j \|p^{(j)}\|_1 \leq \max_j \|p^{(j)}\|_1$. This shows that there always exist a {\em deterministic} behaviour of $\gexcl$ that attains the maximum of $S$. Therefore, the maximum of $S$ for classical theories is the size of the largest stable set of $\gexcl$. This is exactly the independence number of $\gexcl$, denoted by $\alpha(\gexcl)$. So we get the inequality $p_1 + p_2 + \cdots + p_n \leq \alpha(\gexcl)$.

\begin{definition}(Independence number)
 Given a graph $G=(V,E)$, Independence number is the size of the largest subset of vertices $S \subseteq V$ such that no pair of vertices in $S$ are connected. Independence number is denoted by $\alpha(G)$.
\end{definition}

\begin{definition}
 A non-contextuality inequality corresponds to a half-space that contains the set of non-contextual behaviours, that is, 
\be 
\sum_{i \in [n]} p_i \leq \alpha(\gexcl),
\ee
for all $p \in \mathcal{P}_{NC}(\gexcl)$. 
\end{definition}

Interestingly in the quantum setting, one has some additional degrees of freedom to increase this sum. Indeed, let state $u_0$ be a unit vector in a complex Hilbert space $\mathcal{H}$. The event $e_i$ correspond to projecting $u_0$ to a one-dimensional subspace, spanned by a unit vector $u_i \in \mathcal{H}$; the probability that the event occurs is just the squared length of the projection. That is, $p_i = |\braket{u_0}{u_i}|^2$ and $p_1 + p_2 + \cdots + p_n = \sum_{i=1}^n |\braket{u_0}{u_i}|^2$. Now two exclusive events must correspond to projections onto orthogonal vectors, and hence $\braket{u_i}{u_j} = 0$, for all edges $(i,j)$ in $\gexcl$. From Lemma~\ref{startpoint}, $p \in TH(\gexcl)$. Therefore, the optimisation problem we are interested in is 
\be \label{sayma}
\max \sum_i p_i : p \in TH(\gexcl).
\ee
In other words, find a matrix $ X\in \mathbb{S}^{1+n}_+, \text{ with the largest diagonal sum such that } X_{00}=1, \ X_{ii} = X_{0i} \, \forall i \in [n], \ X_{ij}=0, \forall (i,j)\in E\ $. This is precisely the definition of the Lov\'asz theta SDP~\eqref{lovtheta} corresponding to $\gexcl$. The value of this SDP is the famous Lov\'asz theta number $\vartheta(\gexcl)$.
\be\label{lovtheta}
\begin{aligned} 
\vartheta(\gexcl) = \max & \ \sum_{i=1}^n { X}_{ii} \\
{\rm s.t.} & \ { X}_{ii}={ X}_{0i}, \ i\in [n],\\
 & \ { X}_{ij}=0,\ i\sim j,\\
& \ X_{00}=1,\ X\in \mathcal{S}^{n+1}_+.
\end{aligned}
\ee
\noindent Hence we get $p_1 + p_2 + \cdots + p_n \leq \vartheta(\gexcl)$.

\section{Graph-theoretic dimension witnesses}
\label{sec2}


Any Bell or contextuality inequality can be associated to a graph of exclusivity~\cite{CSW}. In this sense, all of them can be studied under the  graph-theoretic framework. While in all previous works one first fixes a (Bell or contextuality) scenario and then looks for dimension witnesses, in this work we investigate the dimension witnesses for graphs (of exclusivity), without fixing a priori any scenario. 


\subsection{Quantum correlations with dimensional restrictions}


In this section we examine from a graph-theoretic perspective the problem of
quantum correlations (aka behaviours) with dimensional restrictions. We use some standard concepts of graph theory and the graph-theoretic approach to correlations introduced in Section~\ref{notation_context}.

\begin{definition}(\textbf{$d$-quantum behaviour for a graph of exclusivity}) A behaviour $p: [n]\to~[0,1]$ corresponding to a graph of exclusivity $\gexcl$, having $n$ vertices, is $d$-quantum if there exists a quantum state $\ket{\psi} \in \mathcal{H}^d$ and non zero projectors $\Pi_1,\ldots, \Pi_n$, belonging to a $d$-dimensional Hilbert space $\mathcal{H}^d$~such that 
\be\label{rankdef}
p_i= \bra{\psi}\Pi_i\ket{\psi},\, \forall i\in [n] \text{ and } \Tr(\Pi_i\Pi_j)=0, \text{ for } i \sim j.
\ee
\end{definition}
We call a quantum realization of the behaviour $p$, the set $\ket{\psi}, \{\Pi_i\}_{i=1}^n \in \mathcal{H}^d$ satisfying \eqref{rankdef}. We denote the set of $d$-quantum behaviours by $\mathcal{P}_{Q}^d(\gexcl)$.

\begin{definition}(\textbf{Orthogonal rank}) The orthogonal rank of a graph $G$, denoted by $R_o(G)$, is the minimum $d$ such that there exists a $d$-dimensional orthonormal representation for $G$. 
\end{definition}

\noindent For example, any orthonormal representation of the $3$-cycle graph of exclusivity must consist of three mutually orthonormal vectors and therefore must be of dimension at least~$3$. Therefore, $R_o(C_3) = 3$. Note that ${\cal P}^d_Q(\gexcl)$ is an empty set for $d < R_o(\gexcl)$. 

Suppose that we are interested in the largest value of the expression $\sum_{i \in [n]} p_i$, as $p$ ranges over the set of $d$-quantum behaviours, that is, the following optimisation problem: 
\be\label{dimbehaviour}
 \max \sum_{i=1}^{n} p_i : p \in \mathcal{P}_{Q}^d(\gexcl).
\ee
Removing the dimensional constraint, the set of quantum behaviours $\mathcal{P}_{Q}(\gexcl)$ becomes the theta body of $\gexcl$, $TH(\gexcl)$ (see Sec.~\ref{notation_context}). As explained in Eq.~(\ref{sayma}), maximizing the $\ell_1$ norm of $p$ over the theta body is equivalently given by the Lov\'asz theta SDP. Therefore, for all $d \geq R_o(\gexcl)$, problem in Eq~\eqref{dimbehaviour} with the dimensional constraint is equivalently expressed by the following rank constrained version of the Lov\'asz theta SDP: 
\be\label{theta:primalrank}
\begin{aligned} 
\vartheta^d(\gexcl) = \max & \ \ \sum_{i=1}^n {X}_{ii} \\
\text{ subject to} & \ \ {X}_{ii}={ X}_{0i}, \ \ 1\le i\le n,\\
& \ \ { X}_{ij}=0, \ \ i\sim j,\\
& \ \ X_{00}=1,\ \ X\in \mathcal{S}^{1+n}_+, \\
& \ \ \text{rank}(X) \leq d.
\end{aligned}
\ee 

 More concretely, using the same arguments as in Lemma~\ref{startpoint}, if $p \in \mathcal{P}_{Q}^d(\gexcl)$ is optimal for \eqref{dimbehaviour} and $ \{\ket{u_i}\bra{u_i}\}_{i=0}^n \in \mathbb{C}^d$ is a quantum realization of $p$ ( where $\ketbra{u_0}$ refers to the quantum state where as $\ketbra{u_i}$ for $1 \leq i \leq n$, refers to the $n$ projectors), then the Gram matrix of the vectors $\ket{u_0},\braket{u_0}{u_1}\ket{u_1},\ldots,\braket{u_0}{u_n}\ket{u_n}$ corresponds to an optimal solution for~\eqref{theta:primalrank} of rank at most~$d$. 
Conversely, for any optimal solution $X={\rm Gram}(\ket{u_0},\ket{u_1},\ldots,\ket{u_n})$, with $u_i \in \mathbb{C}^d$, of the SDP \eqref{theta:primalrank}, 
the realization $\{{\ket{u_i}\bra{u_i} / \|\ket{u_i}\bra{u_i}\|}\}_{i=0}^n$ is optimal for \eqref{dimbehaviour}. The equivalence fails to hold for $d < R_o(\gexcl)$, due to the inverse norm factor in the above line, since $\|u_i\|=0$ for at least one $i$. This is because otherwise $\{u_i/\|u_i\|\}_{i=1}^n$ is a valid orthonormal representation for $\gexcl$ of dimension $d < R_o(\gexcl)$, violating the definition of orthogonal rank. The quantities $\vartheta^1(\gexcl),\vartheta^2(\gexcl), \ldots, \vartheta^{R_o(\gexcl)-1}(\gexcl)$ are still well-defined but they do not seem to have any physical relevance in this context.

\medskip

On the other hand, we are also interested in the minimum dimension in which the Lov\'asz theta bound can be achieved. 
\begin{definition}(\textbf{Lov\'asz rank}) The Lov\'asz rank of a graph $G$,  denoted by $R_L(G)$, is the minimum $d$ for which $\vartheta^d(G) = \vartheta(G)$.
\end{definition}

\noindent By definition, $R_L(G) \geq R_o(G)$. $R_L(G)$ can be sometimes much smaller than the number of vertices of $G$. The following lemma due to Barvinok~\cite{Barvinok1995} gives an upper bound on $R_L(G)$.

\begin{lemma}(\textbf{Barvinok bound}) \label{barvinok}
There exists an optimal solution of $X^*$ of the following SDP 
\be
\begin{aligned}
\max : & \ \ \tr(CX) \\
\mathrm{s.t.} & \, \tr (A_i X) = b_i, \quad \forall i = 1,2,\ldots, m \\ 
& X \succeq 0,
\end{aligned}
\ee
with rank $r$ satisfying the inequality $r(r+1)/2 \leq m$.
\end{lemma}

\noindent For the Lov\'asz theta SDP, the number of linear constraints is $m = 1 + |V| + |E|$. Hence $R_L(G) \leq \frac{1}{2} \left(\sqrt{8(|V| + |E|)+9}-1 \right)$. To summarise, we have the following relationships:
\be 
\vartheta^{R_o(\gexcl)}(\gexcl) \leq \vartheta^{R_o(\gexcl)+1}(\gexcl) \leq \cdots \leq \vartheta^{R_L(\gexcl)}(\gexcl) = \vartheta(\gexcl).
\ee
This suggests a way to lower bound the dimension of the underlying quantum system that violates a certain dimension restricted non-contextuality inequality. More formally, a violation of the inequality $\sum_i p_i \leq \vartheta^d(\gexcl)$, where $p \in {\cal P}_Q(\gexcl)$, implies that the underlying quantum system must have dimension at least $d+1$. We shall refer to the operator in such a dimension restricted non-contextuality inequality as a {\cal dimension witness} for dimension $d+1$.  

\medskip

Finally, we note an equivalent way to compute the dimension restricted Lov\'asz theta, which we define as:
\be
\begin{aligned}\label{Prog_2}
\mathcal{\theta}^d(G) = &\max_{\{v_i \in \mathbb{C}^d\}_{i= 1}^n} \lambda_{max}\left(\sum_{i= 1}^n\ketbra{v_i}\right) \\
& \mathrm{s.t.} \; \braket{v_i}{v_i} = 1, \forall i \in [n]\\ 
& \mathrm{and} \; \braket{v_i}{v_j} = 0, i\sim j.
\end{aligned}
\ee

\begin{lemma}\label{Lmax}
$\theta^d(G) = \vartheta^d(G)$.
\end{lemma}

\begin{proof}
\textsf{($\geq$ direction)} Let $X$ be a solution of SDP. Let $X = VV^{\dagger}$ and the rows of $V$ be $v_i \in \mathbb{C}^{d}$ for $0\leq i \leq n$. Let $\tilde{v_i} = v_i /\|v_i\|$. Clearly, $\tilde{v_i}$ satisfies the constraints in~(\ref{Prog_2}). Now observe that
\be
\begin{aligned}
  \theta^d(G) \geq & \lambda_{max}\left(\sum_{i=1}^{n} \ketbra{\tilde{v_i}}\right) = \max_{v: \|v\|=1} \sum_{i=1}^n |\braket{v}{\tilde{v_i}}|^2 \\
  &\geq \sum_{i=1}^n |\braket{v_0}{\tilde{v_i}}|^2 = \sum_{i=1}^n |\braket{v_i}{\tilde{v_i}}|^2 = \sum_{i=1}^n \braket{v_i}{v_i} \\
  &= \vartheta^d(G). 
\end{aligned}
\ee

\textsf{($\leq$ direction)} Let $\{v_i \in \mathbb{C}^{d}\}_{i=1}^n$ be a an optimal solution of $\theta^d(G)$ and let $v_0$ be the eigen-vector of $\sum_{i= 1}^n\ketbra{v_i}$ corresponding to the largest eigenvalue. Now construct a $(n+1) \times d$ matrix $V$, with $V_0 = v_0$, the first row of $V$ and $V_i = \braket{v_i}{v_0}v_i$, for all $i \in [n]$. Let $X = VV^\dagger$. Firstly, we note that it satisfies all the constraints of the SDP. Now observe that 
\be
\begin{aligned}
 \vartheta^d(G) & \geq \tr(X) -1 \\ 
 &= \sum_{i=1}^n \braket{v_i}{v_i}|\braket{v_i}{v_0}|^2 \\
  &= \sum_{i=1}^n |\braket{v_i}{v_0}|^2 \\
  &= \lambda_{max} \left(\sum_{i=1}^{n} \ketbra{v_i} \right)  \\
  &= \theta^d(G). 
\end{aligned}
\ee
\end{proof}\qedhere


\subsection{Finding low rank solutions: Heuristic approach}
\label{heuristics}


Unfortunately, \emph{rank-constrained} SDPs are \emph{NP}-hard problems and hence they are computationally intractable. An easy way to see this is that the NP-hard \textsf{Max-Cut} problem with weight matrix $W$ can be expressed as the following rank one restricted SDP: 
\be
\begin{aligned} 
\max \ \ &\frac{1}{2}\Tr (W X) \\
\text{s.t.}\ & {X}_{ii}= 1, \forall i,\\
&X \succeq 0, \\
&\text{rank}(X) = 1.
\end{aligned}
\ee

\noindent Because of this restriction, it seems unlikely that given a non-contextuality inequality and a dimension $d$, one can efficiently compute the value $\vartheta^d(\gexcl)$ and find a quantum realisation of dimension $d$ that achieves the bound. Nevertheless, it is important to find such low dimensional quantum realisations which at least violate the classical bound $\alpha(\gexcl)$. For this purpose, we provide a heuristic technique (algorithm \ref{algo:heuristic}) to compute a lower bound on the $d$ dimensional restricted quantum value and find the corresponding $d$-dimensional quantum realisations. 

\medskip

\begin{algorithm}[H]
 \SetKwInOut{Input}{input}
  \SetKwInOut{Output}{output}

  \Input{Graph $G$ having $n$ nodes, dimension $d$, number of iterations \texttt{k}}
  \Output{A lower bound to $\vartheta^d(G)$}
  \vspace{0.2cm}
Generate a random matrix $W \in \mathbb{R}^{(n+1)\times (n+1)}$\;
\textit{iter} = 1\;
\While{iter $< \texttt{k}$}{
Minimise $\tr((W-I_{n+1})X)$, subject to $X \succeq 0$, $X_{00} = 1$, $X_{ii} = X_{0i}$ for all $i$ and $X_{ij} = 0$ for all $i \sim j$\; 
Obtain optimal $X$ for the above SDP\;
Minimise $\tr(XW)$, subject to $I_{n+1} \succeq W \succeq 0$, $\tr(W) = n+1-d$\;
Obtain optimal $W$ from the above SDP \; 
\textit{iter} = \textit{iter} + 1\;
}
\caption{Heuristics using SDPs.} 
\label{algo:heuristic}
\end{algorithm}

\medskip

The algorithm is adapted from an approach to solving rank constrained problems given in Chapter 4 of ~\cite{dattorro2005convex}. The reference gives a heuristic algorithm for producing low rank solutions to feasibility SDP of the form: 
\vspace{-0.5cm}

\be \label{genranksdp}
\begin{aligned} 
\text{Find} \ \ & G \in \mathcal{S}^N_+ \\
\text{ s.t. } & G \in \mathcal{C}\\
&\text{rank}(G) \leq d,
\end{aligned}
\ee

\noindent where $\mathcal{C}$ is a convex set. Instead of solving this non-convex problem directly, they suggest to solve a couple of SDPs~\eqref{ranksdp1} and \eqref{ranksdp2} iteratively, until the following stopping criteria is met. After a particular iteration, let $G^*$ and $W^*$ be the optimal solution of the SDPs ~\eqref{ranksdp1} and \eqref{ranksdp2} respectively. The loop is stopped if $\langle G^*, W^*\rangle = 0$. Let us see why. Note that the eigenvalues of $W^*$ lie in the closed interval $[0,1]$ and they sum up to $N-d$. This implies that at least $N-d$ of its eigenvalues are non-zero, that is, rank$(W^*) \geq N-d$. This, along with the fact that $\langle G^*, W^*\rangle = 0$, implies that rank$(G^*) \leq d$. Since $G^*$ is a solution of the first SDP, it must also satisfy the conditions $G^* \in \mathcal{C}$ and $G^* \in \mathcal{S}^N_+$. Thus $G^*$ is a solution of SDP~\eqref{genranksdp}. However, note that there is no guarantee that the stopping criteria will be met. 
\begin{multicols}{2}

\be \label{ranksdp1}
\begin{aligned} 
\min_G \ \ &\langle G,W \rangle \\
\text{ s.t. } & G \in \mathcal{C}\\
& G \in \mathcal{S}^N_+.
\end{aligned}
\ee

\be \label{ranksdp2}
\begin{aligned} 
\min_W \ \ &\langle G,W \rangle \\
\text{ s.t. } & \tr(W) = N - d \\
& I_N \succeq W \succeq 0. 
\end{aligned}
\ee
\end{multicols}

In our case, the SDP~\eqref{theta:primalrank} is more general in the sense that it also involves optimising an objective function. Thus we include the objective function of the Lov\'asz theta SDP, $\tr(X)$, as an extra additive term to the objective function of the first SDP~\eqref{ranksdp1}. Besides this, the main idea of Algorithm~\ref{algo:heuristic}, is same as in the feasibility SDP case - to solve two SDPs iteratively. The first SDP tries to satisfy all the Lov\'asz theta SDP constraints, while the second SDP tries to restrict the rank of the solution $X$ to the desired value. The algorithm is made to run for a predefined number of iterations, $\texttt{k}$. In the end of the program, if the final $X$ and $W$ are such that $\langle X,W \rangle = 0$, then the solution $X$ is indeed a feasible solution to SDP~\eqref{theta:primalrank}. If not, we restart the program. We find that this heuristic works pretty well in practice and enables us to find low rank solutions to the Lov\'asz theta SDP. Taking a Gram decomposition of the solution matrix $X$ allows us to compute the $d$ dimensional quantum realisations. 

\medskip

Note that Algorithm~\ref{algo:heuristic} only outputs a lower bound for $\vartheta^d(G)$ and is not directly used to find dimension witnesses (which would require an upper bound). However one may expect to guess this upper bound by running this algorithm several times (by taking the maximum among all the runs). This idea allows us to find candidate graphs for which we can find dimension witnesses and prove the upper bound theoretically. In fact, in Sec.~\ref{sec4Qites}, we describe a family of graphs, which can be used as dimension witnesses, which was found precisely by the same logic using Algorithm~\ref{algo:heuristic}.



\subsection{Examples}


To demonstrate the usefulness of the tools introduced, we apply them to two of graphs which are relevant in the literature on contextuality. For each graph, we report the lower bounds on the rank constrained Lov\'asz theta values for different dimensions obtained with the algorithm introduced before\footnote{A MATLAB implementation of the code using the SDPT3 solver, can be found \href{https://www.dropbox.com/sh/595q05xpo7wfzpd/AAC8jvuprr-C-DTcJccxl6fea?dl=0}{here}.} and discuss why the results are interesting. 


\begin{figure}[H]
	\centering
	\includegraphics[width=0.4\textwidth]{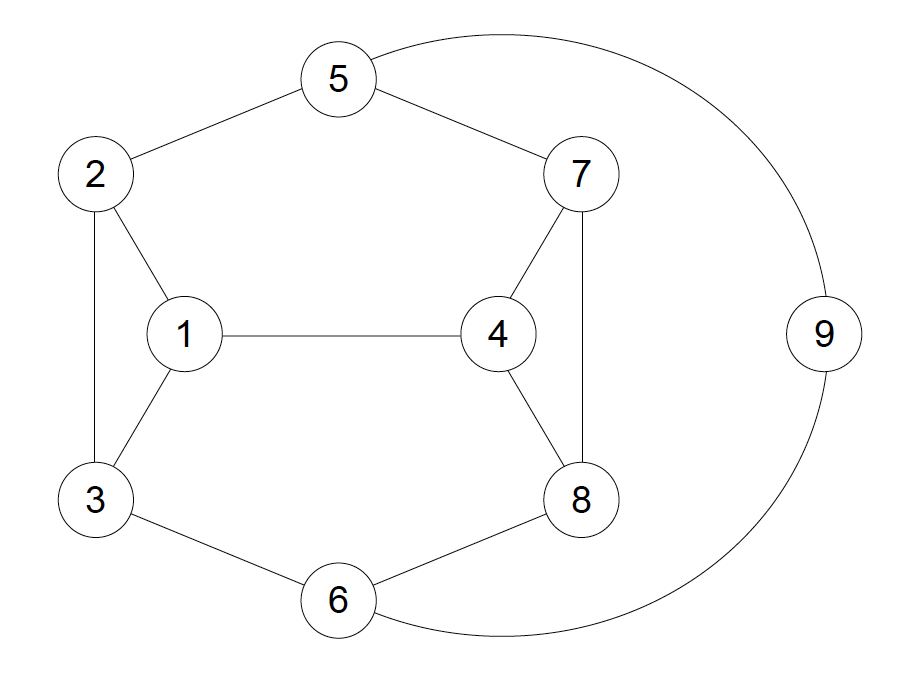}
	\caption{$G_1$ graph: The $9$-vertex graph $G_1$ was used in \cite{Dagomirgraph} to illustrate the notion of almost state-independent contextuality.}
	\label{dagograh}
\end{figure}


\subsubsection{Almost state-independent contextuality}


The earliest proof of state-independent quantum contextuality by
Kochen and Specker \cite{Kochen:1967JMM} required 117 three-dimensional real projective measurements. Since then, the number of projective measurements needed to demonstrate state-independent contextuality has been drastically reduced to thirteen over the years \cite{cabello1996bell, yu2012state}.  The paper by Yu and Oh suggested a test to reveal state-independent contextuality with only thirteen
projectors \cite{yu2012state}. Later, a computer-aided proof confirmed that it is impossible to demonstrate state-independent contextuality
with less than thirteen measurements \cite{cabello2016quantum}. Thus, any test of contextuality with less than thirteen projective measurements
would fail to exhibit contextuality for at least a few quantum states. The $9$-vertex graph $G_1$ in Fig.~\ref{dagograh} is a part of the original proof of the Kochen-Specker theorem \cite{Kochen:1967JMM} and has been used in \cite{Dagomirgraph} to illustrate the concept of ``almost state-independent'' contextuality. The almost state-independent non-contextuality inequality is given by,
\be \label{eq: Dag_ineq}
\sum_{i \in [n]} p_i \leq 3,
\ee
with the events satisfying the exclusivity relation given by the graph in Fig.~\ref{dagograh}. In reference \cite{Dagomirgraph}, authors showed that the non-contextuality inequality in \eqref{eq: Dag_ineq} is saturated by a three dimensional maximally mixed state and violated by every other choice of three-dimensional preparation, for an appropriate choice of measurement settings. Since the non-contextuality inequality in \eqref{eq: Dag_ineq} is violated for every quantum state, except maximally mixed state, it exemplifies the concept of almost state-independent contextuality. For details, refer to \cite{Dagomirgraph}.  As one can see, the non-contextual bound for the aforementioned non-contextuality inequality is given by its independence number, $\alpha(G_1) = 3$ \cite{CSW}. In addition, $R_o(G_1) = 3$ and $R_L(G_1) \leq 4$.  
Our calculations lead to the following results: 
\begin{center}
\begin{tabular}{ |c|c|c| } 
 \hline
$d =$ & $3$ & $4$ \\ 
 \hline
$\vartheta^d(G_1) \geq$ & $3.333$ & $3.4706=\vartheta(G_1)$ \\ 
 \hline
\end{tabular}
\end{center}
The authors of \cite{Kochen:1967JMM,Dagomirgraph} used this graph to illustrate state-independent and almost state-independent in $d=3$, respectively. From numerics, we know that there exists a rank 4 solution which achieves the Lov\'asz theta number and it would be interesting to show that $R_L(G_1) = 4$. Also, numerical evidence suggests that $\vartheta^3(G_1) \leq 3.333$, however we do not have theoretical proof. If we assume $\vartheta^3(G_1) \leq 3.333$, it would mean that any experimental value $> 3.333$ will certify that the underlying dimension is greater than $3$.



\begin{figure}[H]
 \centering
  \includegraphics[width=0.4\textwidth]{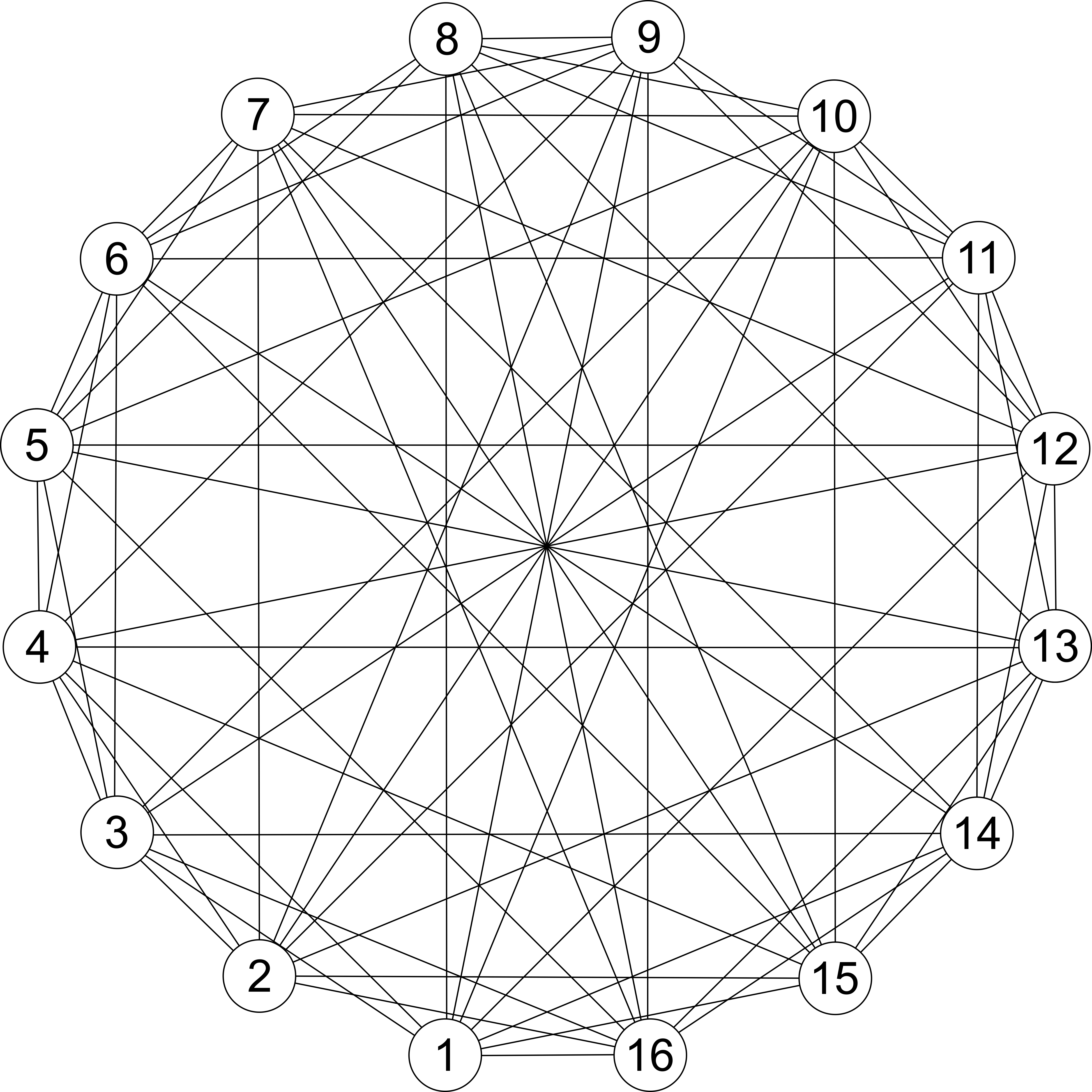}
  \caption{$G_2$ graph: the $16$-vertex graph $G_2$, is the graph of exclusivity corresponding to the $16$ events in the Bell operator of Mermin's tripartite Bell inequality. The aforementioned tripartite Bell inequality can be used to self-test the $3$-qubit GHZ state.}
  \label{mermingrah}
\end{figure}


\subsubsection{Mermin's Bell inequality}


We discuss an $n$-partite Bell inequality (for odd $n\geq 3$ ), known as Mermin's Bell inequality  \cite{Mermin90}, the interest of which is based on the fact that the Bell operator
\begin{equation}
S_n = \frac{1}{2 i} \left[\bigotimes_{j=1}^n (\sigma_x^{(j)}+i \sigma_z^{(j)}) - \bigotimes_{j=1}^n (\sigma_x^{(j)}-i \sigma_z^{(j)})\right],
\end{equation}
where $\sigma_x^{(j)}$ is the Pauli matrix $x$ for qubit $j$, has an eigenstate with eigenvalue $2^{(n-1)}$ . In contrast, for local hidden-variable (LHV) and noncontextual hidden-variable (NCHV) theories,
\begin{equation}
\langle S_n \rangle \overset{\scriptscriptstyle{\mathrm{LHV, NCHV}}}{\le} 2^{(n-1)/2}.
\end{equation}
The aforementioned inequality thus demonstrates the fact that there is no
limit to the amount by which quantum theory can surpass the limitations imposed by local hidden variable theories (or non-contextual hidden variable theories). We are interested in the tripartite case, i.e. for $n=3$,
\begin{equation} \label{eq: mermin_3}
 \langle \sigma_z^{(1)} \otimes \sigma_x^{(2)}  \otimes \sigma_x^{(3)} \rangle + \langle
\sigma_x^{(1)} \otimes \sigma_z^{(2)} \otimes \sigma_x^{(3)}\rangle  +
\langle \sigma_x^{(1)} \otimes \sigma_x^{(2)} \otimes \sigma_z^{(3)} \rangle-
\langle \sigma_z^{(1)} \otimes \sigma_z^{(2)} \otimes \sigma_z^{(3)} \rangle \leq 2.
\end{equation}
The tripartite inequality in \eqref{eq: mermin_3} can be used to self-test a $3$- qubit GHZ state \cite{kaniewski2017self}. One can study the aforementioned inequality via the graph approach introduced  in \cite{CSW}. The $16$-vertex graph $G_2$ in Fig.~\ref{mermingrah} is the graph of exclusivity corresponding to the $16$ events in the Bell operator of Mermin's tripartite Bell inequality~\cite{mermin_cabello}. In this case, $\alpha(G_2) = 3$, $R_o(G_2) = 4$, and $R_L(G_2) \leq 7$. Our calculations give
\begin{center}
\begin{tabular}{ |c|c|c|c|c| } 
 \hline
$d =$ & $4$ & $5$ & $6$ & $7$ \\ 
 \hline
$\vartheta^d(G_2) \geq$ & $3.414$ & $3.436$ & $3.6514$ & $4=\vartheta(G_2)$ \\ 
 \hline
\end{tabular}
\end{center}
Further if we can show that these lower bounds are tight, then one can use these inequalities as dimension witnesses. It is also interesting to note that the Lov\'asz theta can be achieved in $d=7$, since achieving it in the three-party, two-setting, two-outcome Bell scenario requires $3$ qubits and thus $d=2^3=8$.


\subsection{Quantum dimension witnesses for arbitrary dimensions : the family of Qites}
\label{sec4Qites}


It was realised \cite{Kochen:1967JMM} that achieving Kochen-Specker contextuality requires quantum dimension of at least $3$. A simple proof of this is provided in the following Lemma. 
 
\begin{lemma}\label{gleason}
$\vartheta^2(\gexcl) = \alpha(\gexcl)$.
\end{lemma}
\begin{proof}

For this proof we use the definition of the restricted Lov\'asz theta number from (\ref{Prog_2}). We need to show that, if we restrict ourselves to $2$ dimensional vectors, then the restricted Lov\'asz theta number is at most the independence number of the graph. Firstly note that if the graph has an odd cycle ($>1$), then it cannot have orthonormal representation in 2 dimensions. Thus we consider only bipartite graphs. Furthermore, assume that $\gexcl$ is connected. If it is not connected, apply the same arguments as follows, to each connected component and then note that the independence number of the graph is the sum of the independence number of its connected components. For a connected bipartite graph its bi-partition is unique and for $\gexcl$, let them be denoted as $V$ and $V'$. The key observation is that for any unit vector $\ket{v}$ in $\mathbb{C}^2$, there exists a unique (up to a unit complex number $e^{i \theta}$) vector $\ket{v^{\perp}}$ that is orthogonal to $\ket{v}$. This implies that if we assign a unit vector $v \in \mathbb{C}^2$ to a vertex in $V$ then all the vectors in $V$ must be of the form $e^{i \theta} \ket{v}$, for some $\theta \in [0,2\pi]$, whereas all vectors in $V'$ must be of the form $e^{i \theta} \ket{v^{\perp}}$. This implies that the cost of the orthonormal representation is at most $\lambda_{\max} \left(\sum_{i \in V} \ketbra{v} + \sum_{i \in V'} \ketbra{v^{\perp}}\right) = \max \{|V|, |V'| \} = \alpha(\gexcl)$. \qedhere
\end{proof}

\medskip

To look for more interesting dimension witnesses for arbitrary higher dimensions we define a family of graphs parameterised by integers $k \geq 2$, called \emph{k-Qite}\footnote{The reason for is that they resemble kites. However the name kite is already reserved for another family of graphs.}. 
\begin{definition}
A $k$-Qite graph has $2k+1$ vertices, $v_1,v_2,\ldots, v_{2k+1}$, with the first $k$ vertices forming a fully connected graph. Vertex $v_i$ is connected to vertex $v_{i+k}$, for all $1\leq i \leq k$. Vertex $v_{2k+1}$ is connected to vertices $v_{k+i}$, for all $1\leq i \leq k$. 
\end{definition}

\noindent Note that the first member of the family, that is $k=2$, is just the $C_5$ graph (see Fig.~\ref{fig:2qite}). This is one of the most well studied graphs in the field of contextuality since it is the smallest graph for which the Lov\'asz theta number is strictly greater than the independence number. The corresponding non-contextuality inequality is the famous {\em KCBS} inequality~\cite{KCBS}. The graph corresponding to $k=3$ is shown in Fig.~\ref{fig:3qite}. 


\begin{figure}
  \centering
  \begin{minipage}{0.45\textwidth}
    \centering
    \includegraphics[width=0.75\textwidth]{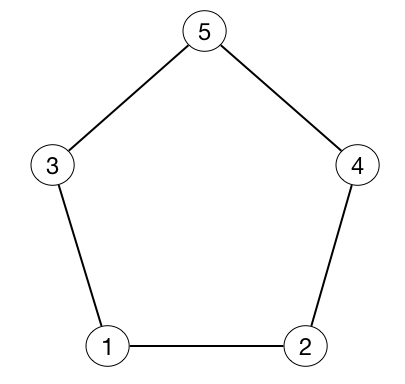} 
    \caption{$2$-Qite $\equiv$ $C_5$, where \hspace{2cm} $\alpha(C_5) = 2, \vartheta(C_5) = \sqrt{5} \approx 2.2361$}.
    \label{fig:2qite}
  \end{minipage}\hfill
  \begin{minipage}{0.45\textwidth}
    \centering
    \includegraphics[width=0.9\textwidth]{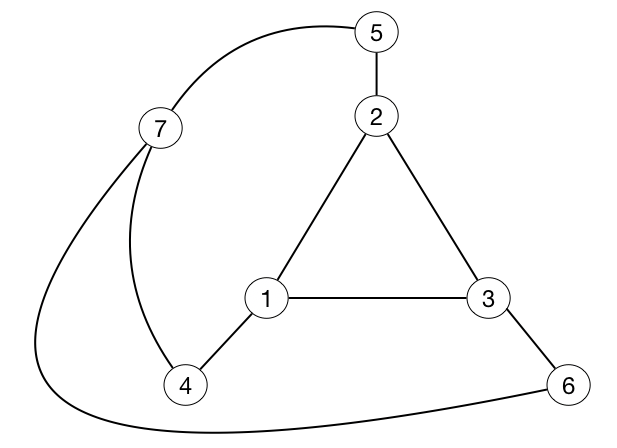} 
    \caption{$3$-Qite, where $\alpha(3\text{-Qite}) = 3, \vartheta(3\text{-Qite}) \approx 3.0642$}.
    \label{fig:3qite}
  \end{minipage}
\end{figure}


\begin{lemma}\label{qiteind}
The independence number of the $k$-Qite graph is $k$.
\end{lemma}
\begin{proof}
Partition the set of the vertices into three sets: $S_1 = \{v_1,v_2,\ldots, v_{k}\}$, $S_2 = \{v_{k+1},v_{k+2},\newline \ldots, v_{2k}\}$ and $S_3 = \{v_{2k+1}\}$. Firstly note that since none of the vertices in $S_2$ are connected to each other, the independence number is at least $|S_2| = k$. Since every vertex in $S_1$ is connected to each other, there can be at most one vertex from $S_1$ in a maximal independent set. However, the inclusion of a vertex from $S_1$, say $v_i$ in the maximal independent set would imply the vertex $v_{k+i}$ cannot be included simultaneously in the maximal independent set. Similarly inclusion of $v_{2k+1}$ implies that one cannot have any vertex of $S_2$ in the maximal independent set. Hence the lemma follows.  
\end{proof} 

\begin{theorem}
$R_o(\emph{k-}Qite) = k$, for all $k \geq 3$.
\end{theorem}
\begin{proof}
Consider the vertex partitioning as in Lemma~\ref{qiteind}. Since vertices in $S_1$ form a $k$-complete graph, we have $R_o(\text{k-}Qite) \geq k$. Now we show that there exists an orthonormal representation in dimension $k$ for all $\text{k-}Qite$ graphs with $k \geq 3$. Depending of the parity of $k$, we give an explicit construction for the orthonormal representation. \newline
\textbf{When $k$ is odd:} For all the vertices in $S_1$, assign the standard vectors $e_i$ in a $k$-dimensional Hilbert space to vertex $v_i$, for $i \in [$k$]$. Assign the vector $\frac{1}{\sqrt{k}}(1,1,\ldots,1)$ to vertex $v_{2k+1}$. Now consider the vertices $v_{k+i}$ in $S_2$, for $i \in [$k$]$. For vertex $v_{k+i}$ to be orthogonal to vertex $v_{i}$, the vector for $v_{k+i}$ must have $0$ in the $i^{th}$ position. Let the magnitude of the remaining entries of the vector be $\frac{1}{\sqrt{k}}$. Since $k$ is odd, the number of entries with non-zero (also equal) magnitude is even. Setting, half of them randomly to negative sign, makes it orthogonal to the vector $v_{2k+1}$. Hence, in this case, all orthonormality constraints are satisfied. 
\newline
\noindent \textbf{When $k$ is even:} Assign the vectors to all the vertices in $S_1$ in the same way as in the odd $k$ case. Set the vector corresponding to vertex $v_{2k+1}$ as $\frac{1}{\sqrt{k-1}}(0,1,1,\ldots,1)$. Except vertex $v_{k+1}$, set all the rest of the vertex in $S_2$ in the same way as in the odd $k$ case. Note that this establishes orthogonality of vertex $v_{k+i}$ with $v_{2k+1}$ for all $2 \leq i \leq k$. Vertex $v_{k+1}$ is then set such that its first entry is $0$ (to make it orthogonal to $v_1$) and is orthogonal to $v_{2k+1}$. There are many such vectors which would satisfy these conditions. For example, set $v_{k+1}$ as $\frac{1}{\sqrt{(k-2)(k-1)}}(0,1,1,\ldots,1,2-k)$ to conclude the proof.
\end{proof}

In order to propose dimension witnesses, we want to find upper bounds on the dimension restricted Lov\'asz theta number corresponding to the {\em Qite} family. For $k=2$, Lemma~\ref{gleason} already gives us the required bound of $2$. We now generalise the Lemma for the {\em Qite} family. 

\begin{theorem}
$\vartheta^k(\emph{k-}Qite) \leq k$, for all $k \geq 2$.
\end{theorem}
\begin{proof}
We use the $\theta^d(G)$ definition of rank restricted Lov\'asz theta for the proof, see Lemma~\ref{Lmax}. $\vartheta^k(\emph{k-}Qite) = \max_{\{v_i\}} \lambda_{max} \left(\sum_{i=1}^{2k+1} \ketbra{v_i} \right)$, where $\ket{v_i} \in \mathbb{C}^k$ is a $k$-dimensional quantum state corresponding to the vertex $v_i$, such that $\braket{v_i}{v_j} = 0$, whenever vertices $v_i$ and $v_j$ share an edge. Since, the first $k$ vectors must form an orthogonal basis (as they form a $k$-complete graph), one can suppose that $\ket{v_i} = e_i$ (the standard basis vector), for $1\leq i \leq k$, without loss of generality. This is because there will always exist a unitary $U$, that can rotate any orthonormal basis to the standard basis. Note that this unitary rotation on all the vertices, gives us another set of orthonormal representation of the graph with the same cost, that is, 
\be
\begin{aligned}
  \lambda_{max}\left(\sum_{i=1}^{2k+1} \ketbra{v_i}\right) &= \lambda_{max}\left(U\left(\sum_{i=1}^{2k+1} \ketbra{v_i}\right)U^{\dagger}\right) = \lambda_{max}\left(\sum_{i=1}^{2k+1} U\ketbra{v_i}U^{\dagger}\right). 
\end{aligned}
\ee

\noindent Since $ \sum_{i=1}^{k} \ketbra{v_i} = \mathbb{I}$, we are required to show that $\lambda_{max}\left(\sum_{i=k+1}^{2k+1} \ketbra{v_i}\right) \leq k-1$. Note that setting the first $k$ vectors to the standard basis vectors also implies that the $i^{th}$ component of $\ket{v_{k+i}}$ is $0$, for $1 \leq i \leq k$. Next, observe that $\ket{v_{2k+1}}$ is orthogonal to $\ket{v_{k+i}}_{i=1}^k$ and so $\lambda_{max}\left(\sum_{i=k+1}^{2k+1} \ketbra{v_i}\right) \leq \max \{\lambda_{max}\left(\sum_{i=k+1}^{2k} \ketbra{v_i}\right), 1\}$. Hence it suffices to show that $\lambda_{max}\left(\sum_{i=k+1}^{2k} \ketbra{v_i}\right) \leq k-1$.

Let $M \in \mathbb{C}^{k \times k}$ be the matrix whose $i^{th}$ row is $\ket{v_{k+i}}^{\T}$, for $i \in [k]$. Note that $M^{\dagger}M = \sum_{i=k+1}^{2k} \ketbra{v_i}$. Also, observe that $M$ has the property that it's diagonal is all zero and it's rows are all normalized to 1 in $\ell_2$-norm. We shall now bound the largest eigenvalue of $M^{\dagger}M$. We make use of Gershgorin's circle theorem which states that given a complex square matrix $A \in \mathbb{C}^{n \times n}$, it's eigenvalues (which may be complex) lie within at least one of the $n$ Gershgorin discs, that is a closed disk in the complex plane centered at $A_{ii}$ with radius given by the row sum $r_i = \sum_{j\neq i } |A_{ij}|$ for $1 \leq i \leq n$. Since $M_{ii} = 0$ for all $i$, 
\be\max_{x: \|x\|=1}\|Mx\|_2 = |\lambda_{max}(M)| \leq \max_{k+1\leq i \leq 2k} \, \|\ket{v_{i}}\|_1 \leq \sqrt{k-1} \max_{k+1\leq i \leq 2k} \, \|\ket{v_{i}}\|_2 = \sqrt{k-1}, 
\ee
where the second inequality follows from the fact that the $\ell_1$-norm of a vector $v$ is at most $\sqrt{\dim(v)}$ times it's $\ell_2$-norm. Finally putting everything together, 
\be
\lambda_{max}(M^{\dagger}M) = \max_{x: \|x\|=1} x^{\dagger}M^{\dagger}Mx = \max_{x:\|x\|=1} \|Mx\|_2^2 \leq (\sqrt{k-1})^2 = k-1.
\ee
\end{proof}

On the other hand, one can verify that $\vartheta(\emph{k-}Qite) > k$, for any $k > 1$, by solving the Lov\'asz theta SDP for the $\emph{k-}Qite$ graph numerically. This gives us the following corollary.  

\begin{corollary}
Violating the non-contextuality inequality $\sum_i p_i \leq k$ where $p \in {\cal P}_{Q}(\emph{k-}Qite)$, implies that the underlying quantum realisation must have dimension at least $k+1$. 
\end{corollary}


\section{Conclusion}
\label{disc}


In this work, we have introduced a novel approach to quantum dimension witnessing in scenarios with one preparation and several measurements (examples of them are Kochen-Specker contextuality and Bell nonlocality scenarios). Our approach is based on graphs which represent the relations of exclusivity between events. Each graph can be realized in different scenarios, and there is always a (specific Kochen-Specker contextuality) scenario for which all quantum behaviours for the graph can be realized. The virtue of our approach is precisely that we do not need to fix any scenario. Instead, we explore the features of abstract graphs for dimension witnessing. Here, we have introduced all the necessary tools to identify graph-based dimension witnesses, and we have illustrated their usefulness by showing how famous exclusivity graphs in quantum theory hide some surprises when re-examined with our tools and how one can construct simple dimension witnesses for any arbitrary dimension. Arguably, however, the main interest of our results is that they can be extended in many directions, connected to multiple problems, and applied to in different ways. Here we list some of possible future lines of research: 
\begin{itemize}
	\item Identifying graph-theoretic dimension witnesses for specific Bell and Kochen-Specker contextuality scenarios.
	\item Using previous knowledge in graph theory for finding useful quantum dimension witnesses. For example, there are graphs for which the ratio of Lov\'asz theta number to independence number is quite large, i.e., $\frac{\vartheta(G)}{\alpha(G)} \gg 1$ \cite{Feige1997RandomizedGP, amaral2015maxcontext}. This indicates situations where the quantum vs classical advantage is highly robust against imperfections. Therefore, dimension witnesses based on such graphs could be useful for certification tasks on, e.g., noisy intermediate-scale quantum devices \cite{preskill2018quantum}. 
	
	\item For the purpose of noise robust dimension witnesses, one may also use a weighted version of graphs (corresponding to a weighted non-contextuality inequality). As an example, for our family of  $\emph{k-}Qite$ graphs, one can consider a weight vector given by $w=(1,1,\ldots,1,k-1)$, where more weight is given to the $(2k+1)^{th}$ vertex of $\emph{k-}Qite$. Note that the weighted independence number of this weighted graph is still $k$. However numerically solving the weighted Lov\'asz theta for this graph suggests $\vartheta(\emph{k-}Qite,w) - \alpha(\emph{k-}Qite,w)> 0.26$ for all $k \geq 3$. For large $k$ this difference converges to $\approx 1/3$. However note that since for large $k$, the ratio $\frac{\vartheta(\emph{k-}Qite,w)}{\alpha(\emph{k-}Qite,w)} \approx 1$, this approach is still not noise robust.

	\item Implementing graph-theoretic quantum dimension witnesses in actual experiments. 
	\item Obtaining the classical memory cost \cite{kleinmann2011memory,CGGX18} for simulating graph-theoretic dimension witnesses and identifying quantum correlations achievable with low-dimensional quantum systems but requiring very-high dimensional classical systems.
	\item Extending the graph-theoretic framework to classical dimension witnessing.
	\item Developing a general graph-theoretic framework to analyse and unify different approaches to dimension witnessing.
\end{itemize}


\section*{Acknowledgments}


The authors thank Zhen-Peng Xu for valuable comments on the ar$\chi$iv version and suggesting the use of weighted graphs for increasing the quantum-classical gap as described in the conclusions. The authors also thank Antonios Varvitsiotis for helpful discussions. We also thank the National Research Foundation of Singapore, the Ministry of Education of Singapore for financial support. This work was also  supported by \href{http://dx.doi.org/10.13039/100009042}{Universidad de Sevilla} Project Qdisc (Project No.\ US-15097), with FEDER funds, \href{http://dx.doi.org/10.13039/501100001862}{MINECO} Projet No.\ FIS2017-89609-P, with FEDER funds, and QuantERA grant SECRET, by \href{http://dx.doi.org/10.13039/501100001862}{MINECO} (Project No.\ PCI2019-111885-2). 


\bibliographystyle{alpha}
\bibliography{dimwit}




	
\end{document}